\documentclass[lettersize,journal]{IEEEtran}
\usepackage{amsmath,amsfonts}
\usepackage{algorithm,ulem}
\usepackage{array}
\usepackage[caption=false,font=normalsize,labelfont=sf,textfont=sf]{subfig}
\usepackage{textcomp}
\usepackage{stfloats}
\usepackage{url}
\usepackage{verbatim}
\usepackage{graphicx}
\usepackage{cite}
\usepackage{booktabs}
\usepackage{multirow}
\usepackage{makecell}
\usepackage{xcolor}
\usepackage{ulem}
\usepackage{hyperref}
\hypersetup{hidelinks}

\usepackage{float}
\usepackage{amsthm}
\usepackage{amssymb}
\usepackage{enumitem}
\usepackage{bbm}
\usepackage{mathtools}
\usepackage{algpseudocode}
\newtheorem{theorem}{Theorem}

\newtheorem{lemma}{Lemma}
\newtheorem{assumption}{Assumption}
\newtheorem{definition}{Definition}

\newtheorem{remark}{Remark}
\usepackage{threeparttable}

\begin{document}

\title{Optimal Battery Placement in Power Grid}


\author{Ruotong Sun$^{1}$, Ermin Wei$^{2}$ and Lihui Yi$^{3}$%
\thanks{$^{1}$Ruotong Sun is with Industrial Engineering and Management Sciences Department, Northwestern University, Evanston, IL 60208, USA.}
\thanks{$^{2}$Ermin Wei is with Electrical and Computer Engineering Department and Industrial Engineering and Management Sciences Department, Northwestern University, Evanston, IL 60208, USA.}
\thanks{$^{3}$Lihui Yi is with the Department of Electrical and Computer Engineering, Northwestern University, Evanston, IL 60208, USA.}}




\maketitle

\begin{abstract}
We study the optimal placement of an unlimited-capacity battery in power grids under a centralized market model, where the independent system operator (ISO) aims to minimize total generation costs through load shifting. The optimal battery placement is not well understood by the existing literature, especially regarding the influence of network topology on minimizing generation costs. Our work starts with decomposing the Mixed-Integer Linear Programming (MILP) problem into a series of Linear Programming (LP) formulations. For power grids with sufficiently large generation capacity or tree topologies, we derive analytical cost expressions demonstrating that, under reasonable assumptions, the weighted degree is the only topological factor for optimal battery placement. We also discuss the minor impact of higher-order topological conditions on tree-topology networks. To find the localized nature of a single battery's impact, we establish that the relative cost-saving benefit of a single battery decreases as the network scales. Furthermore, we design a low-complexity algorithm for weakly-cyclic networks. Numerical experiments show that our algorithm is not only approximately 100 times faster than commercial solvers but also maintains high accuracy even when some theoretical assumptions are relaxed.


\end{abstract}


\section{Introduction}
\IEEEPARstart{I}{n} the past decades, distributed Generation (DG) emerges as a decentralized energy solution, with Distributed Energy Systems (DES) gaining widespread global support \cite{nadeem2023distributed}.  
A key challenge in the economic operation of DES is managing the temporal mismatch between electricity demand and the availability of low-cost generation, as system-wide generation costs often fluctuate throughout the daily cycle \cite{reichelstein2015time}.


Battery Energy Storage Systems (BESS) have emerged as a pivotal technology to mitigate this issue by performing temporal energy arbitrage: storing energy when it is inexpensive and discharging it during high-cost periods. This load-shifting capability can significantly reduce reliance on costly generation and enhance overall economic efficiency \cite{das2018overview, helal2018optimal}.

Existing literature investigates the reliability and efficiency of BESS from diverse perspectives. To minimize total power losses on transmission lines, \cite{Tang17} developed a continuous tree model to determine the optimal placement of storage within radial networks. \cite{Qin19, qin2016submodularity} incorporated capital and installation costs through a discrete optimization framework, deriving conditions under which the value function of storage placement is submodular. Further, market interactions involving BESS have been explored. For instance, \cite{QinLi19} proposed a storage investment game to model the interactions between system operators and storage investors. Research has also examined optimal scheduling of BESS for phase balancing \cite{helal2018optimal, sun2015distributed}. However, they fails to specify battery location \cite{pinthurat2023techniques}, a critical factor in practical implementation.

The economic value of a BESS, however, is not realized in a vacuum; it is fundamentally constrained by the physical laws and topology of the power grid. A battery's ability to perform cost-saving arbitrage is limited by network congestion even in the case of no transimission loss and arbitrarily large battery capacity. Its placement is therefore a critical factor in maximizing its system-wide benefits.

Many studies have investigated numerical methods for optimizing energy storage placement and sizing across various applications \cite{wogrin2014optimizing, pandvzic2014near, sardi2017multiple, ahlawat2023optimal, ali2024optimization, ghaffari2022optimal}. Specifically, \cite{sardi2017multiple} examined the benefits of energy storage within distribution systems, using numerical search methods to identify optimal placement, sizing, and operational characteristics. Aiming to optimize multiple objectives, including energy costs, \cite{ali2024optimization} proposed a large-scale multi-objective evolutionary algorithm. 

However, theoretical investigations on optimal battery placement problems remain limited, as they present substantial theoretical challenges \cite{wu2023distributed}: for each battery placement locations and capacities, we should also determine the charging and discharging profile of batteries, involving multiple dimensions of states and decision variables. Under deterministic demand, \cite{thrampoulidis2015optimal} is among the first to address the optimal placement problem theoretically, formulating it as an optimal storage placement problem for load shifting. However, the applicability of their results is limited to systems where generators are connected to the network via a single link. Bose et al. \cite{Bose14} linked the locational marginal value of storage to the locational marginal price at each bus under stochastic demand under the assumption of extremely small storage capacities, which limits its applicability. More later, Wu et al. \cite{wu2023distributed} analyzed a star network with a central generator under stochastic demand, but their findings are constrained by the network topology, making it difficult to generalize the results to other interconnected network structures. Recently, Han et al. \cite{han2024regularized} addressed stochastic demand by employing a regularized Mixed-Integer Programming model within a two-stage stochastic programming framework. Their objective was to minimize total load shedding and excess grid power. This choice of objective, while crucial to their analysis, represents a different formulation of the operational problem, focusing on physical balance instead of cost. Consequently, their results do not apply directly to scenarios focused on minimizing generation costs. Overall, the influence of network topology on optimal battery placement to reduce generation cost is an essential yet underexplored consideration.

In this study, we consider a regulated and centralized market, where the ISO has full knowledge of the deterministic demands and generation costs, and schedules the electricity supplied by each generator. We explore the optimal placement of batteries within the power grid from a theoretical perspective, utilizing the battery's load-shifting capability to reduce total generation costs. Specifically, this study aims to answer the following questions: Given a single available battery with infinite capacity, where is the optimal placement in a lossless power system? In addition, how does the topology of the power grid impact optimal placement? 

Our main contributions include the following: 
\begin{itemize} 
\item We decompose the original MILP problem into a series of LP subproblems.
\item We derive an analytical form for the operational costs under specific settings and propose a low-complexity algorithm to efficiently solve for these costs in weakly-cyclic power grids.
\item Through theoretical analysis, we provide insights into the impact of network topology on optimal battery placement. 
\end{itemize}

In this paper, we use the notation $\mathbb{R}^n_{\geq 0}$ and $\mathbb{R}^n_{> 0}$ to represent $n$-dimensional vectors with non-negative and strictly positive components, respectively. Moreover, $\mathbbm{1}$ denotes the vector with all components equal to 1, dimensioned appropriately to the context of use. 

The remainder is organized as follows. We introduce our model and mathematically formulate the optimal battery placement problem and simplify it in Section \ref{sec: model}. We then present our main theoretical results in Section \ref{sec: theo analysis}. Based on this analysis, Section \ref{algorithm} proposes an efficient algorithm for weakly-cyclic topologies. The performance of this algorithm is evaluated through extensive numerical experiments in Section \ref{ne}. Finally, Section \ref{conclusion} summarizes our findings and key insights. Due to page limitations, detailed proofs and supplementary examples are deferred to the appendix.


\section{Model}\label{sec: model}


\subsection{DC Power flow}
In this section, we introduce our power flow model for a general smart grid. Consider a connected power network $G=(N,E)$ with $n$ buses and $m$ transmission lines. Similar to \cite{sun2015distributed}, each bus has a controllable generator and a load. The load profiles often exhibit cyclical behavior \cite{QinLi19}, so it is sufficient to consider a finite horizon in our model and denote the entire period by $T$. 

We use $\mathcal{T}=\{1, 2, \dots, T\}$ to denote the set of time periods. For each time $t$ in $\mathcal{T}$, the generation profile is represented by the vector $g_t \in \mathbb{R}^n$, and the load profile (or demand) is denoted by the vector $d_t \in \mathbb{R}^n_{\ge 0}$. The generation cost is modeled as a linear function \cite{Bose14}, where the cost coefficient is given by the vector $c_t \in \mathbb{R}^n_{> 0}$. In this study, as we focus on a power system utilizing a single energy source within a localized region, we introduce a homogeneous cost assumption. 
\begin{assumption}[Homogeneous cost]\label{HomoCost}
We set the cost (coefficient) to be homogeneous across the buses within the same time period but varies between different time periods, i.e., $c_t=c(t)\mathbbm{1}$, where $c(t)\in \mathbb{R}_{>0}$ and $c(t_1) \neq c(t_2) \footnotemark, \forall\ t_1 \neq t_2$.
\end{assumption}
\footnotetext{Our main results (cf. Theorem \ref{optimal charging profile}, \ref{thm2}) rely on an ordering of costs across different time periods, which simplifies the discussion. Should the situation arise where two or more time periods have identical cost, our results remain valid by simply resolving this equivalence. This can be achieved through an arbitrary tie-breaking mechanism.}
Each generator is subject to a maximum generation capacity at any time $t$, denoted by the vector $\tilde{g}_t \in \mathbb{R}^n$. 



Moreover, this paper adopts the DC power flow approximation \cite{Yuanzhang}. We represent the power flow equation as $p_t = B\theta_t, \ t=1, \dots, T,$ where $p_t \in \mathbb{R}^n$ denotes the net power injection, $B \in \mathbb{R}^{n \times n}$ is the admittance matrix, and $\theta_t \in \mathbb{R}^n$ is the voltage phase angle at time $t$. 
A practical scenario is considered in which each line has a maximum capacity. We concatenate these limits into a vector $f \in \mathbb{R}^m$. The power flowing through the $m$ lines is constrained by:
\begin{align*}
    -f \le A\theta_t \le f, \quad t=1,...,T,
\end{align*}
where $A\in \mathbb{R}^{m \times n}$ is the incidence matrix weighted by lines' susceptance. We refer the reader to \cite{Bose14} for the detailed description of matrix $A$ and $B$. 



\subsection{Battery Dynamic}
Our goal is to find the optimal placement of a battery within a power grid. For notational convenience, we let $u_t \in \mathbb{R}^n$ represent the battery control profile at each time $t$ across all the buses. This vector essentially captures the energy charged to or discharged from the battery across all buses at any given time $t$. In this paper we use the subscript $i,t$ to indicate the corresponding value at bus $i$ and time $t$. Considering a battery at bus $b$, note that $u_{b,t} > 0$ represents that the battery discharges and behaves as a generator, while $u_{b,t} < 0$ represents the battery charges and behaves as a load at time $t$. The state of charge (SoC) at any time $t$ is $x_{t} \in \mathbb{R}^n_{\geq 0}$. We assume a stylized dynamic process \cite{qin2016submodularity} as follows,
\begin{align*}
    x_{t+1} = x_{t} - u_{t}, \quad t=1,...,T-1,
\end{align*}
The initial SoC is assumed to be $x_{1}=z\in \mathbb{R}^n_{\ge 0}$ which is a design parameter. To preserve the cyclic behavior, we require $x_{T+1}=x_{1}$.
The constraints with respect to the battery dynamic are assembled as follows. 
\begin{align}\label{2}
    \begin{dcases}
        z\otimes\mathbbm{1}+UH \ge 0,\\
        U\mathbbm{1} = 0,
    \end{dcases}
\end{align}
where $\otimes$ is the outer product between two vectors, $U = [u_{1}, u_{2}, ..., u_{T}]\in \mathbb{R}^{n\times T}$, $\mathbbm{1}$ has $T$ dimension, and $H \in \mathbb{R}^{T \times T}$ is a upper triangular matrix constructed as follows,
\begin{align*}
    H_{ij} = -1,\quad \forall\ i \le j.
\end{align*}
Note that the constraint $z\in \mathbb{R}^n_{\ge 0}$ is embedded in \eqref{2}.
These constraints imply that at each time at all locations, the battery level is nonnegative, while also guaranteeing that the net change in the battery level over the time horizon $T$ is zero, thereby preserving the cyclic behavior of the system.

Since the battery is the most expensive component in a microgrid \cite{wongdet2023optimal} especially their installation, it is more cost-effective to install a battery with a sufficiently large storage capacity at a single bus from the outset, rather than deploying multiple smaller batteries across many buses or dealing with frequent replacements, which can significantly increase operating costs. In this paper, we consider the scenario where a single battery with a large enough capacity is installed within the network.

We introduce a binary vector $\sigma \in \{0,1\}^n$, with component $i$ represents if the battery is placed at bus $i$, i.e., 
\begin{align*}
    \begin{dcases}
        \sigma_{i} = 1, &\text{if battery is placed at bus }i \\
        \sigma_{i} = 0, &\text{otherwise}.
    \end{dcases}
\end{align*}

Since only one battery is introduced, it naturally follows that $\mathbbm{1}'\sigma = 1$. Suppose that the battery is placed at bus $b$, only the component corresponding to the $b$-th bus in $u_t$ can be nonzero, which means $u_{i,t} = 0$ for all $i \neq b$. Hence, we construct a series of ``if-then" constraints,
\begin{align}\label{eq1}
    \begin{dcases}
        u_{i,t} \leq M\sigma_{i}, &\forall i\in N, t\in \mathcal{T} \\
        u_{i,t} \geq -M\sigma_{i}, &\forall i\in N, t\in \mathcal{T},
    \end{dcases}
\end{align}
where $M$ is a large enough constant. These constraints ensure that $u_{i,t}$ can take nonzero values if and only if $\sigma_i = 1$, which indicates that a battery is installed at the $i$-th bus.

\subsection{Optimal Battery Placement}
The optimal battery placement can be modeled as a large-scale MILP with the objective minimizing the total generation cost. The installation cost is assumed to be constant and is therefore excluded from the formulation. 
\begin{align*}\tag{$P0$}\label{P0}
    \min_{g_t, \theta_t, U, z, \sigma} \quad 
    & \sum_{t=1}^{T} c(t)\mathbbm{1}'g_t \\
    s.t. \quad
    & g_t + u_t - d_t = B\theta_t, && \forall t\in \mathcal{T} \\
    & -f \le A\theta_t \le f,        && \forall t\in \mathcal{T} \\
    & 0 \le g_t \le \tilde{g}_t,   && \forall t\in \mathcal{T} \\
    & -\sigma M \leq u_t \leq M\sigma, && \forall t\in \mathcal{T} \\
    & z\otimes\mathbbm{1}+UH \ge 0, &&\\
    & U\mathbbm{1} = 0,            && \\
    & \mathbbm{1}'\sigma = 1,      && \\
    & \sigma \in \{0,1\}^n.        &&
\end{align*}

After attaining the optimal solution, the position of non-zero element in $\sigma$ represents the optimal location for the battery. 
\subsection{Multi-period economic dispatch}
Although the general MILP is NP-hard, we can decompose the problem into $n$ multi-period economic dispatch problems with one for each battery placement choice. The following theorem reformulate \eqref{P0} as a two-stage optimization problem.

\begin{theorem}\label{model simplification}
We can equivalently reformulate \eqref{P0} a two-stage optimization. In the first stage, we determine the optimal battery placement bus
\begin{align*}
    b^* \coloneqq \underset{b=1,...,n}{\arg\min} \quad J(b),
\end{align*}
where $J(b)$ is the optimal operational cost of the second-stage problem
\begin{align*}\tag{$P$}\label{P}
    \min_{g_t, \theta_t, [U]_b} \quad &\sum_{t=1}^{T} c(t)\mathbbm{1}'g_t\\ 
    s.t. \qquad&g_t+u_t-d_t = B\theta_t, && \forall t\in \mathcal{T} &&\lambda_t \in \mathbb{R}^n\\
    &A\theta_t \ge -f, && \forall t\in \mathcal{T} &&\underline{\mu}_t \in \mathbb{R}_{\ge 0}^m\\
    &A\theta_t \le f, && \forall t\in \mathcal{T} &&\overline{\mu}_t \in \mathbb{R}_{\ge 0}^m\\
    &g_t \ge 0, && \forall t\in \mathcal{T} &&\underline{\beta}_t \in \mathbb{R}_{\ge 0}^n\\
    &g_t \le \tilde{g}_t,  && \forall t\in \mathcal{T} &&\overline{\beta}_t \in \mathbb{R}_{\ge 0}^n\\
    &[U]_b\mathbbm{1} = 0. && &&\eta \in \mathbb{R}
\end{align*}
This is a multi-period economic dispatch problem given the battery placement $b$. We use $[\cdot]_b$ to represent the $b$-th row of matrix $U$. Besides, $u_{i,t} = 0$ for all $t$ and $i \neq b$ is taken as the parameter. 

The dual variables in \eqref{P} are introduced to correspond to the constraints of the optimization problem. Specifically, $ \lambda_t \in \mathbb{R}^n $ is the dual variable associated with the power balance constraints and is also known as the Locational Marginal Price (LMP)\cite{schweppe2013spot}, representing the marginal cost of supplying one additional unit of power at each bus at time $ t $. The variables $ \underline{\mu}_t \in \mathbb{R}_{\geq 0}^m $ and $ \overline{\mu}_t \in \mathbb{R}_{\geq 0}^m $ correspond to the line flow limits, while $ \underline{\beta}_t \in \mathbb{R}_{\geq 0}^n $ and $ \overline{\beta}_t \in \mathbb{R}_{\geq 0}^n $ enforce the generation capacity limits. Finally, $ \eta \in \mathbb{R} $ is associated with the cyclic battery operation constraint $ [U]_b\mathbbm{1} = 0 $, ensuring energy balance over the time horizon.
\end{theorem}

The theorem follows from the fact that our battery has sufficient capacity, allowing the initial SoC to be shifted higher to ensure that $ z\otimes\mathbbm{1}+UH \ge 0 $ is always satisfied, while keeping the battery control and generation profile unchanged. Consequently, this constraint is redundant in our model.

In the remainder of the paper,  we use $b$ to denote the battery placement bus.

\subsection{Maxflow Decomposition}
Although the generation profiles and battery controls are temporally coupled in \eqref{P}, further decomposition can be made due to the homogeneous cost coefficient. In this part, we introduce the \textit{Maxflow decomposition}, including \textit{Maxflow problems}, followed by two key theorems regarding the optimal charging profile and costs. Both theorems highlight the relationship between Maxflow problems and the optimal solution of \eqref{P}, demonstrating why Maxflow problems are crucial to our analysis and how their objective values are embedded within the optimal solution of \eqref{P}.

Let $W$ be the feasible region for Maxflow problems.
$$
W = \left\{ (g_t, u_t, \theta_t) \left| 
\begin{aligned}
    &g_t + u_t - d_t = B\theta_t, \\
    &-f \leq A\theta_t \leq f, \\
    &0 \leq g_t \leq \tilde{g}_t.
\end{aligned}
\right.\right\}
$$
We begin with the definition of Maxflow problems, which includes the \textit{maximum inflow} and \textit{outflow} determined for each time period $t\in \mathcal{T}$ and any battery placement bus $b\in N$.
\begin{definition}[Maximum inflow]
Maximum inflow is the maximum power flow into the battery placement bus at a period, which is:
\begin{align*}\tag{$MIF$}\label{MIF}
    F_{b,t}^{\text{in}}\coloneqq \max_{g_t, \theta_t, u_t} \quad &\sum_{i\neq b}(g_{i,t} - d_{i,t})\\
    s.t.\quad&(g_t, \theta_t, u_t) \in W.
\end{align*}
\end{definition}
\begin{definition}[Maximum outflow]
The maximum outflow is the maximum power flow out of battery placement bus at a period, which is:
\begin{align*}\tag{$MOF$}\label{MOF}
    F_{b,t}^{\text{out}}\coloneqq \max_{g_t, \theta_t, u_t} \quad &\sum_{i\neq b}(d_{i,t} - g_{i,t})\\
    s.t.\quad&(g_t, \theta_t, u_t) \in W.
\end{align*}
\end{definition}

Collectively, we refer to \eqref{MIF} and \eqref{MOF} as {\it Maxflow problems}, which are similar to, but distinct from, the usual minimum-cost flow problem \cite{ahuja1988network}, a classic problem in operations research. The latter typically includes only the flow conservation constraint, i.e., Kirchhoff's first law, to ensure that the total flow into a node is equal to the total flow out. In contrast, our Maxflow problems here incorporate both Kirchhoff’s first and second laws. While Kirchhoff’s first law ensures flow conservation, the second law, which states that the voltages around any closed loop sum to zero, is also reflected in the first constraint of $W$.

To ensure that \eqref{MIF} and \eqref{MOF} have feasible solutions and are well-defined, we make the following assumption on $\tilde{g}_t$, which ensures that every generator has enough generation capacity to meet their local demand.
\begin{assumption}[Sufficient generation capacity]\label{Sufficient generation assumption}
The generation capacity of each bus exceeds its corresponding demand at all times, that is, $\tilde{g}_{i,t} > d_{i,t}, \forall i \in N, t\in \mathcal{T}$.
\end{assumption}

Based on this assumption, the next lemma shows that our Maxflow problems are well defined, which means that they are feasible and the inflow and outflow are positive (consistent with their physical interpretation). 
We have Assumptions \ref{HomoCost} and \ref{Sufficient generation assumption} hold for the rest of the paper.

\begin{lemma}\label{positive maxflow}
    Both $F_{b,t}^{\text{in}}$ and $F_{b,t}^{\text{out}}$ are strictly positive.
\end{lemma}
It is important to distinguish between the Maxflow and the charging profile $u_{b,t}$, as the latter also incorporates power generation and consumption at the battery placement bus $b$. The relationship between them is given by $ u_{b,t} \in [-F_{b,t}^{\text{in}} - \tilde{g}_{b,t} + d_{b,t}, F_{b,t}^{\text{out}} + d_{b,t}] $. Lemma \ref{positive maxflow} guarantees that this interval is always well-defined and non-empty. Building on this result, we derive the following theorem, which establishes that the charging profile in \eqref{P} can be determined by solving the Maxflow problems.

\begin{theorem}\label{optimal charging profile}
Let $t^{(i)}$ be the $i$-th lowest-cost period. In the optimal solution, set $k = \max\{i \in \mathcal{T}|u_{b,t^{(i)}} < F_{b,t^{(i)}}^{\text{out}}+d_{b,t^{(i)}}\}$, where $t^{(k)}$ is called the adjustment period.
Then the optimal charging profile can only be:
\begin{align*}
    \begin{dcases}
        u_{b,t} = -F_{b,t}^{\text{in}}-\tilde{g}_{b,t} + d_{b,t}, \quad\text{if } t = t^{(1)}, \cdots, t^{(k-1)}\\
        u_{b,t^{(k)}}\in \left[-F_{b,t^{(k)}}^{\text{in}}-\tilde{g}_{b,t^{(k)}} + d_{b,t^{(k)}},F_{b,t^{(k)}}^{\text{out}} + d_{b,t^{(k)}}\right),\\
        u_{b,t} = F_{b,t}^{\text{out}} + d_{b,t},\quad \text{if } t = t^{(k+1)}, \cdots, t^{(T)}
    \end{dcases}
\end{align*}
where $u_{b,t^{(k)}}$ is determined by $[U]_b\mathbbm{1} = 0$.
\end{theorem}
The intuition behind this theorem is straightforward. Since the objective in \eqref{P} is to minimize the total oprational cost, the battery should be charged as much as possible when prices are low, and the stored power should be released to meet the grid demand when prices are high.

Essentially, it states that the optimal charging profile in \eqref{P} follows a bang-bang control. To see this, without loss of generality, we assume that the marginal cost strictly increases over time. Then, according to Theorem \ref{optimal charging profile}, the battery should be charged before the adjustment period and then discharged.

Once the optimal charging profile for each time period is determined, the optimal cost of placing the battery at bus $b$ can be derived analytically, as shown in the next theorem.
\begin{theorem}\label{thm2}
We can use the following lexicographic order to determine the adjustment period, which corresponds to the $k$-th lowest cost period, where
$$
    k=\begin{dcases}
    1 & \text{if } \tilde{g}_{b,t^{(1)}}+F_{b,t^{(1)}}^{\text{in}} \geq \sum_{i=2}^T F_{b,t^{(i)}}^{\text{out}} + \sum_{t} d_{b,t}\\
    \vdots\\
    m & \text{if } \sum_{i=1}^{m-1}\tilde{g}_{b,t^{(i)}}+\sum_{i=1}^{m-1}F_{b,t^{(i)}}^{\text{in}} < \sum_{i=m}^T F_{b,t^{(i)}}^{\text{out}} + \sum_{t} d_{b,t}\\
    & \sum_{i=1}^{m}\tilde{g}_{b,t^{(i)}}+\sum_{i=1}^{m}F_{b,t^{(i)}}^{\text{in}} \geq \sum_{i=m+1}^T F_{b,t^{(i)}}^{\text{out}} + \sum_{t} d_{b,t}\\
    \vdots\\
    T & \text{if } \sum_{i=1}^{T-1}\tilde{g}_{b,t^{(i)}}+\sum_{i=1}^{T-1}F_{b,t^{(i)}}^{\text{in}} < F_{b,t^{(T)}}^{\text{out}} + \sum_{t} d_{b,t}
    \end{dcases}
$$   
The corresponding cost of \eqref{P} will be:
\begin{align}\label{general optimal cost}
    &\sum_{t=1}^T c(t)\mathbbm{1}'g_t
    = \sum_{t=1}^T\sum_{i=1}^n d_{i,t}c(t) - \sum_{t=1}^T\left(F_{b,t}^{\text{in}} + \tilde{g}_{b,t} - d_{b,t}\right) \nonumber\\ 
    &(c(t^{(k)})-c(t))^+ -\sum_{t=1}^T (F_{b,t}^{\text{out}} + d_{b,t}) (c(t^{(k)})-c(t))^-
\end{align}
\end{theorem}
The theorem indicates that the larger maximum inflow (outflow) during lower-cost (higher-cost) periods leads to a smaller $k$, which means that more power can be ``moved" from lower- to higher-cost periods, reducing the operational cost.

We define the weighted degree of any bus $i$ as $\sum_{j\in \mathcal{N}(i)}f_{ij}$. This theorem implies that if $F_{b,t}^{\text{in}}$ and $F_{b,t}^{\text{out}}$ are solely dependent on the weighted degree of $b$, the only topological factor that affects the cost is the weighted degree of the battery bus.

\section{Theoretical Analysis}\label{sec: theo analysis}
In this section, based on Maxflow decomposition, we find that in some settings the problem has an analytical solution. In particular, we can solve \eqref{MIF} and \eqref{MOF} analytically in such cases, where the relation to the topology is also clear.
\subsection{Power Grid without Congestion}
We first consider the scenario in which the power grid does not experience congestion at any time. In this setting, we have the following theorem.
\begin{theorem}\label{small tilde g}
    When the power gird has no congestion at any time, the battery only discharges at the highest cost period, and the optimal cost is uniform in the battery’s location.
\end{theorem}
This implies that battery placement is irrelevant; one may place the battery at any bus and achieve the same optimal cost.

\textbf{Remark.} Although the proof uses KKT conditions, the solution can also be understood via a Maxflow formulation, whence $F_{b,t}^{in} = \sum_{i\neq b}\left(\tilde{g}_{i,t}-d_{i,t}\right)$. Note that since the battery only discharges in the highest cost period, the coefficient of $F_{b,t}^{out}$ in the optimal cost is zero, shown in \eqref{general optimal cost}.

The next lemma clarifies what we mean by ``no congestion."
\begin{lemma}[Extremely small generation capacity]\label{suff cond no conges}
    There exists a uniform bound $\varepsilon>0$, if $\max_{i\in N, t\in \mathcal{T}}\{\tilde{g}_{i,t}-  d_{i,t}\}<\varepsilon$, then the power grid remains uncongested at all times.
\end{lemma}
Intuitively, when the maximum generation capacity $\tilde{g}_{i,t}$ is very close to the demand $d_{i,t}$, there is insufficient flexibility for large power transfers, so no line becomes congested.

\subsection{Sufficiently large generation capacity}
In following parts, we consider scenarios where the power grid experiences congestion. To this end, we impose the following assumption throughout the remainder of the analysis.
\begin{assumption}[Limited transmission]\label{Congestion assumption}
For all $t$ and $(i,j) \in E$, the line capacity $f_{ij}$ is strictly less than the minimum of demands on the connected buses, i.e., $f_{ij} < \min\{d_{i,t}, d_{j,t}\}$.
\end{assumption}
Intuitively, there is less room for flexibility in the transmission route in that case. Although this assumption does not guarantee that congestion will occur, it creates conditions where congestion is more likely. Moreover, it simplifies the analysis: Facilitates the analysis of KKT (see Theorem~\ref{large tilde g}), and greatly simplifies \eqref{MOF} (see Lemmas~\ref{Maximum outflow in tree topology}, \ref{Maximum outflow in weakly-cyclic topology with homogeneous line capacity}).

Next, following this section's scope, we now assume that $\tilde{g}_t$ is sufficiently large.
\begin{assumption}\label{condition for ``large enough"}
$\tilde{g}_{j, t}\geq f_{ij}+d_{j,t} \ \forall t\in \mathcal{T},\ i\in N,\ j\in \mathcal{N}(i).$
\end{assumption}

Under LP theory, a solution of \eqref{P} is optimal if and only if the KKT conditions are satisfied. However, directly analyzing the original KKT system is difficult. The primary difficulty stems from the optimality condition $B\lambda_t + A^\prime(\overline{\mu}_t - \underline{\mu}_t) = 0$, which corresponds to the partial derivative of the Lagrangian 
with respect to the phase angle $\theta_t$.

This condition imposes global constraints on power transmission, creating complex interdependencies between distant nodes that hinder analytical progress. To this end, we want to analyze a ``local" power transmission case. We first assume the following.
\begin{assumption}\label{unit admittance}
    All transmission lines have unit admittance.
\end{assumption}
 This yields the matrix identity $B = A^\prime A$. Then we focus on a subset of KKT solutions satisfying a locality condition below (cf. \eqref{Enhanced KKT}, we will justify its feasibility after Theorem \ref{large tilde g}). The following theorem characterizes \eqref{MIF} and \eqref{MOF}.

\begin{theorem}\label{large tilde g}
Under Assumptions~\ref{Congestion assumption}, \ref{condition for ``large enough"} and \ref{unit admittance}, we use superscript $*$ to denote the solution of KKT conditions and impose\begin{equation}\label{Enhanced KKT}
    \left\{\text{(KKT solutions)} \left| A\lambda_t^* + (\overline{\mu}_t^* - \underline{\mu}_t^*) = 0, \forall t\in \mathcal{T}\right.\right\} \neq \emptyset.
\end{equation}
Then, at any time $t \in \mathcal{T}$, $F_{b,t}^{\text{in}} = F_{b,t}^{\text{out}} = \sum_{i\in \mathcal{N}(b)}f_{bi}$.
\end{theorem}
We give a \textit{proof sketch} here. Recall that the LMP at time $t$ is given by the dual variable $\lambda_t$ in \eqref{P}. First note that \eqref{Enhanced KKT} indicates that if at time $t$, two neighboring buses $i$ and $j$ have different LMPs, then the full line capacity $f_{ij}$ is dispatched. Based on this property and the structure of LP we can argue:
\begin{itemize}
    \item At most one time period can have uniform LMPs.
    \item In other periods, all non-battery buses share a common LMP, while the battery bus has a distinct value.
\end{itemize}
From this observation, all the results can be derived. Q.E.D.

As discussed after Theorem \ref{thm2}, in that case, how much the battery can charge or discharge over multiple periods is determined by the weighted degree of the battery bus and independently of higher-order topology characteristics. This ``local" property is due to \eqref{Enhanced KKT}, essentially.

We next address the feasibility of \eqref{Enhanced KKT}. While Assumption \ref{Sufficient generation assumption} guarantees solutions to \eqref{P} and the corresponding KKT system exist (e.g., all buses are self-sufficient, no power transmission), the condition $A\lambda_t^* + (\overline{\mu}_t^* - \underline{\mu}_t^*) = 0$ excludes solutions where $A\lambda_t + (\overline{\mu}_t - \underline{\mu}_t) \in \mathrm{Null}(A^\prime)$. This may render \eqref{Enhanced KKT} infeasible, as demonstrated by a concrete example in the appendix.

We next discuss a few scenarios where \eqref{Enhanced KKT} holds. In a network with tree topology, \eqref{Enhanced KKT} naturally holds because of the absence of cyclic structures, eliminating the need for additional assumptions. In general topologies, a simple sufficient condition for \eqref{Enhanced KKT} is homogeneous line capacities. In this case, we can simply set $\theta_{b,t} = 0$ and $\theta_{i,t} = \Lambda_t$ for all $i \neq b$, where $\Lambda_t$ is a constant depending on time. Then, power transmission only occurs along the battery's adjacent lines. By setting $|\Lambda_t|$ equal to the line capacity for $T-1$ time periods (while adjusting the remaining periods to maintain cyclic SoC), we achieve the maximum power transmission in each period, consequently obtaining the optimal cost solution.

Beyond these two cases, the following theorem provides a more general condition under which \eqref{Enhanced KKT} admits a solution.

\begin{theorem}\label{feasible enhance KKT}
Under Assumptions~\ref{Congestion assumption}, \ref{condition for ``large enough"} and \ref{unit admittance}. Let $D \coloneqq \max_{i} |\mathcal{N}(i)|$ denote the maximum degree in the network. 
If the following conditions hold:
$$
\begin{dcases}
\frac{D-1}{D-2}\max_{{(i,j)}\in E} f_{ij} - \min_{{(i,j)}\in E} f_{ij} \leq \min_{i,t}\frac{d_{i,t}}{D-2},\\
\frac{\max_{{(i,j)}\in E} f_{ij}}{\min_{{(i,j)}\in E} f_{ij}}\leq\min\left\{2, \frac{D}{D-2}\right\},
\end{dcases}
$$
then \eqref{P} has a optimal solution that satisfies \eqref{Enhanced KKT}.
\end{theorem}
These conditions extend the homogeneous case and intuitively constrain both the relative and absolute values of $f_e$. The first inequality bounds the absolute magnitude of line capacities and the difference between the maximum and minimum line capacity, while the second inequality limits the ratio between them.

\subsection{Tree-topology Power Grid}
In this section, we consider tree-topology power grids without unit admittance and sufficiently large generation capacity assumptions. Before the analysis, we first give some notation. Due to the properties of the tree structure, we designate the battery placement bus $b$ as the root of the tree. Its neighboring set is denoted by $\mathcal{N}(b)$, and for each $j\in \mathcal{N}(b)$ we denote by $\mathcal{T}_b^j$ the subtree rooted at $j$ (which includes both buses and lines). For notational simplicity, we let $|\mathcal{T}_b^j|$ denote the number of buses in the subtree whenever no ambiguity arises.

Now, we aim to characterize the battery inflow $F_{b,t}^{\text{in}}$ and the outflow $F_{b,t}^{\text{out}}$. The limited transmission assumption (cf. Assumption \ref{Congestion assumption}) enables us to derive an analytical solution for \eqref{MOF} when the network topology is a tree.
\begin{lemma}\label{Maximum outflow in tree topology}
Under Assumption \ref{Congestion assumption}, $\forall t\in \mathcal{T}$, $b\in N$, we have $F_{b,t}^{\text{out}} = \sum_{i\in \mathcal{N}(b)} f_{bi}$, where $\mathcal{N}(b)$ is the neighbor set of $b$.
\end{lemma}

The next step is to characterize the maximum battery inflow. In the tree topology, this quantity can be computed inductively (shown in Algorithm 1 later). Moreover, we have the following assumption and the subsequent lemma.
\begin{assumption}\label{uniform line capa}
    All transmission lines have a uniform capacity.
\end{assumption}
\begin{lemma}\label{maxflow tree}
In a tree-topology power grid satisfying Assumption \ref{uniform line capa}, $F_{b,t}^{\text{in}} = \sum_{j\in \mathcal{N}(b)}\min\left\{\sum_{i\in\mathcal{T}_b^j}(\tilde{g}_i^{t} - d_i^{t}),f\right\}$. 
\end{lemma}
In the above expression, we have slightly abused the notation by letting $f$ denote the scalar homogeneous line capacity. Then the bang-bang control profile can be determined via Theorem \ref{optimal charging profile}, and the optimal cost can be derived using Theorem \ref{general optimal cost}. In particular, we have the following result.
\begin{theorem}\label{tree cost}
In a tree-topology power grid with Assumption \ref{Congestion assumption} and \ref{uniform line capa} hold, the optimal cost will be:
\begin{align*}
    &\sum_{t=1}^T c(t)\mathbbm{1}'g_t = \sum_{t=1}^T\sum_{i=1}^n d_{i,t}c(t) + \sum_{t=1}^Td_{b,t}(c(t^{(k)})-c(t)) \nonumber\\
    &-\sum_{t=1}^T \tilde{g}_b^t(c(t^{(k)})-c(t))^+ - |\mathcal{N}(b)|f\sum_{t=1}^T |c(t^{(k)})-c(t)|\nonumber\\
    &+\sum_{t=1}^T\sum_{j\in \mathcal{N}(b)}\max\left\{f-\sum_{i\in\mathcal{T}_b^j}(\tilde{g}_i^t - d_i^t),0\right\} (c(t^{(k)})-c(t))^+.
\end{align*}

Here, $k$ retains its previous meaning and is determined by the lexicographic order in Theorem \ref{thm2}.
\end{theorem}

The above result illustrates how the network topology influences the optimal cost. The first term is a constant, while the second and third terms depend on the total demand at the battery bus and the adjustment period $k$. Due to the assumption of homogeneous line capacity, the weighted degree simplifies to $|\mathcal{N}(b)|f$, which is just the coefficient in the fourth term. Also, this term is affected by the choice of $k$. Hence, the effect of the first four terms is: when $k$ remains constant, a higher weighted degree correlates with a lower optimal cost.
The last term is bounded by a modified version of betweenness centrality of the battery placement bus $b$ (see appendix for details). This indicates that a bus with a higher betweenness centrality could lead to a lower overall cost and be preferable.

This theorem illustrates the impact of the graph topology on the optimal cost. While the overall effect is complicated, we have two simplified cases in which the solely impact of $d_b^t$, $\tilde{g}_b^t - d_b^t$, and the degree can be determined. In particular, although betweenness centrality may affect cost, we can identify two specific cases where the only topological effect on optimal cost is the weighted degree, independent of any higher-order topological metrics. 

\textbf{Case 1}. If $\forall j \in \mathcal{N}(b), t\in \mathcal{T}$, there is $f \leq \sum_{i \in \mathcal{T}_b^j} (\tilde{g}_i^t - d_i^t)$, then by Lemma~\ref{maxflow tree} the maximum inflow $F_{b,t}^{\text{in}}$ depends only on the weighted degree, and the last term in Theorem~\ref{tree cost} becomes solely a function of this degree.

\textbf{Case 2}. If $\forall j \in \mathcal{N}(b), t\in \mathcal{T}$, there is $f > \sum_{i \in \mathcal{T}_b^j} (\tilde{g}_i^t - d_i^t)$, then similarly the weighted degree fully captures the impact of the network topology on the optimal cost.

In both cases, when two buses have identical $d^t$ and $\tilde{g}^t - d^t$ values across all time periods, the bus with higher (weighted) degree is considered a superior location for battery placement. Moreover, if two buses possess the same degree in the above two cases, the one consistently showing higher values for both $d^t$ and $\tilde{g}^t - d^t$ throughout all time periods presents a more economical choice. In essence, assuming all other factors are constant, greater demand does not necessarily guarantee a more advantageous placement.

\subsection{Performance bound in general topology}
We now explore the extent of benefits from installing a single battery in the power grid. In this part we show that: in a sufficiently ``congested" network, the cost reduction ratio due to battery integration will diminish to zero as the power grid scales up under mild conditions.

To quantify the performance of introducing a battery, let $\text{opt}^0 \coloneqq \sum_{t=1}^T\sum_{i=1}^n d_{i,t}c(t)$, which means the cost without battery. We assume that the optimal cost given by the solver for \eqref{P} with battery placement at bus $i$, i.e., $b=i$, is $\text{opt}(i)$, and the optimal position is $b^* = \arg\min_{i\in N} \text{opt}(i)$. The performance is evaluated by the non-negative value $\delta = \max_{b\in N}\left(\text{opt}^{0}-\text{opt}(b)\right)/\text{opt}^{0}$. This value quantifies the maximum relative cost reduction achieved through optimal battery placement. The subsequent theorem demonstrates that this performance metric approaches zero as the power grid expands with limited transmission capacity.
\begin{theorem}\label{performance bound}
    In general topology with Assumption \ref{Congestion assumption}, suppose that the time horizon $T$ is finite, and as the power grid grows,
    \begin{align}\label{cond_for_performance_bound}
        \begin{dcases}
            \lim_{n\to\infty}\sum_{i=1}^n d_{i,t}= \infty \quad & \forall t\in \mathcal{T},\\
            \max_{b\in N}\left( d_{b,t} + \sup_{n\geq 1}\sum_{i\in \mathcal{N}(b)} d_{i,t}\right) <\infty \quad & \forall t\in \mathcal{T}.
        \end{dcases}
    \end{align}
    Then $\lim_{n\to \infty}\delta=0$. That is, the performance of introducing the battery will decrease to zero as the power grid scales to an infinitely large.
\end{theorem}
\begin{proof}
Note that if we relax \eqref{P} by removing the constraint $g_{t}\leq \tilde{g}_{t}, \forall t\in \mathcal{T}$. The cost, denoted as $\text{opt}^{\infty}(b)$ will not be greater than the original cost $\text{opt}(b)$. The superscipt $\infty$ means that the relaxed problem is equal to setting $\tilde{g}_{i,t}\to \infty, \forall i,t$. Theorem \ref{thm2} also applies. Hence $k=1$ and it is easy to identify:
\begin{align*}
    \text{opt}^{0} - \text{opt}^{\infty}(b) &= \sum_{t=1}^T (F_{b,t}^{\text{out}}+d_{b,t})(c(t)-c(t^{(1)}))\\
    &\overset{(*)}{\leq} \sum_{t=1}^T \left(\sum_{i\in \mathcal{B}(b)}d_{i,t}\right)(c(t)-c(t^{(1)})).
\end{align*}
$(*)$ comes from Assumption \ref{Congestion assumption}, and $\mathcal{B}(b)\coloneqq \mathcal{N}(b)\cup\{b\}$. So, the performance of introducing a battery is upper-bounded.
\begin{align*}
    \delta &= \max_{b\in N}\frac{\text{opt}^{0}-\text{opt}(b)}{\text{opt}^{0}}\leq \max_{b\in N}\frac{\text{opt}^{0}-\text{opt}^{\infty}(b)}{\text{opt}^{0}}\\
    &\leq \frac{\max_{b\in N} \sum_{t=1}^T \left(\sum_{i\in \mathcal{B}(b)}d_{i,t}\right)(c(t)-c(t^{(1)}))}{\sum_{t=1}^T\left(\sum_{i=1}^n d_{i,t}\right)c(t)},
\end{align*}
where the first inequality comes from $$\min_{b\in N}\text{opt}^{\infty}(b)\leq \min_{b\in N}\text{opt}(b) = \text{opt}(b^*).$$
By the conditions, $\lim_{n\to\infty}\sum_{i=1}^n d_{i,t}= \infty$, i.e., the denominator goes to infinity. Moreover, the numerator is finite:
\begin{align*}
    &\max_{b\in N} \sum_{t=1}^T \left(\sum_{i\in \mathcal{B}(b)}d_{i,t}\right)(c(t)-c(t^{(1)}))\\
    \leq & \max_{b\in N} \sum_{t=1}^T \left(d_{b,t}+\sup_{n\geq 1}\sum_{i\in \mathcal{N}(b)}d_{i,t}\right)(c(t)-c(t^{(1)})) <\infty\\
    \Longrightarrow & \sup_{n\geq 1}  \left\{\max_{b\in N} \sum_{t=1}^T \left(\sum_{i\in \mathcal{B}(b)}d_{i,t}\right)(c(t)-c(t^{(1)}))\right\} < \infty.
\end{align*}
This proves our conclusion.
\end{proof}

One sufficient condition for \eqref{cond_for_performance_bound} is that all the demand is lower bounded and that the number of neighbors of any bus is uniform bounded. The intuition behind the theorem is that, with only one battery and limited line capacities, power regulation is restricted to a local range. As a result, the performance of introducing a single battery deteriorates as the power grid expands.

\section{Optimal battery placement algorithm \\ for weakly-cyclic graphs}\label{algorithm}

We now consider the weakly-cyclic topology, a more general class of graphs than the tree topology, where no edge belongs to more than one cycle \cite{wei2023supply}. This includes trees and rings as special cases. For analytical tractability, we assume homogeneous line capacity $f$ and unit admittance throughout this section. Following the approach used in Lemma~\ref{Maximum outflow in tree topology}, we first characterize the maximum outflow under this topology.
\begin{lemma}\label{Maximum outflow in weakly-cyclic topology with homogeneous line capacity}
Under Assumptions \ref{Congestion assumption}, \ref{unit admittance} and \ref{uniform line capa}, in a weakly-cyclic topology, for any time $t$ and battery location $b$, we have $F_{b,t}^{\text{out}} = |\mathcal{N}(b)| f$.
\end{lemma}
We now turn to the more challenging task of computing the maximum inflow \eqref{MIF}. A key ingredient is the notion of connected components formed by removing a single node from the network.
We take $\mathcal{C}_{v}$ as the set of connected components after removing node $v$.

\begin{lemma}\label{connected edge 1 or 2}
Let $\mathcal{C}$ be one of the connected components in $\mathcal{C}_{i}$, and $\mathcal{E}$ denote the set of edges in the original graph that connect the bus $i$ to $\mathcal{C}$. In a weakly-cyclic graph, we claim that $ |\mathcal{E}| = 1 $ (tree-like link) or $ |\mathcal{E}| = 2 $ (ring-like link).
\end{lemma}

This property enables an efficient recursive procedure (see Algorithm 1). When $|\mathcal{E}|$ exceeds 2, the coupling described by Kirchhoff’s second law becomes more complex. In general topology, analyzing this procedure is challenging, since $|\mathcal{E}|$ can be as large as $n-1$. We leave this for future works. However, with either $|\mathcal{E}|=1$ or $2$, the analysis becomes easy.

Algorithm 1 computes the maximum inflow at each bus by traversing its associated connected components and recursively evaluating their contributions.

\begin{algorithm}[h]
\caption{Weakly-Cyclic Topology Maximum Inflow}
\begin{algorithmic}[1]
\State \textbf{Input:} graph $G = (N, E)$, power grid parameters.
\State \textbf{Output:} $F_{t}^{\text{in}}$: Maximum inflow for each bus.
\State Initialize $F_{t}^{\text{re-in}}$, $F_{t}^{\text{in}}$ as empty nested dictionaries 
\For{each $v \in N$} 
    \State Initialize $F_{t}^{\text{re-in}}[v][\mathcal{C}] \gets -1$ for each $\mathcal{C} \in \mathcal{C}_{v}$ 
\EndFor
\For{each $v \in N$} 
    \For{each $\mathcal{C} \in \mathcal{C}_{v}$} 
        \State $F_{t}^{\text{re-in}}[v][\mathcal{C}] \gets$ \Call{MIF}{$v$, $\mathcal{C}$} 
    \EndFor
    \State $F_{t}^{\text{in}}[v] \gets \sum_{\mathcal{C}\in \mathcal{C}_v} F_{t}^{\text{re-in}}[v][\mathcal{C}]$ 
\EndFor

\end{algorithmic}
\end{algorithm}

The algorithm begins by identifying the connected component for each bus $ v \in N $ and initializing $ F_{t}^{\text{re-in}}[v][\mathcal{C}] $, which represents the maximum power transmission from a connected component $\mathcal{C}$ to bus $ v $. This value is subsequently calculated using the MIF function, followed by an aggregation across all $\mathcal{C} \in \mathcal{C}_v$ to determine the maximum inflow (lines 7–12). 

In the MIF function, we distinguish two cases based on Lemma~\ref{connected edge 1 or 2}. If $|\mathcal{E}| = 1$, to calculate the $MIF(v, \mathcal{C})$, we just need to find the single bus $p$ in $\mathcal{C}$ connected to $v$ and calculate the $MIF(p, \mathcal{C}')$, where $\mathcal{C}'\in \mathcal{C}_p, v\notin \mathcal{C}'$. This leads to a tree-wise recursion. If $|\mathcal{E}| = 2$, we treat the subgraph as a ring and apply a subroutine. The detail of the function MIF and correctness of the algorithm is formally established in the appendix.

\begin{remark}\label{specialtree}
    We remind the reader that the assumptions unit admittance and homogeneous line capacity (Assumptions \ref{unit admittance} and \ref{uniform line capa}) are used only in the case $|\mathcal{E}|=2$ in the MIF function. For a tree, Assumptions \ref{unit admittance} and \ref{uniform line capa} can be removed, since the phase angle in any connected component can be shifted arbitrarily.
\end{remark}
The algorithm computes $F_{t,b}^{\text{in}}$ for each $b\in N$ and $t\in \mathcal{T}$. Moreover, $F_{t,b}^{\text{out}}$ is already obtained by Lemma \ref{maxflow tree}. Taking them in lexicographic order and \eqref{general optimal cost} in Theorem \ref{thm2}, we can calculate the optimal cost of \eqref{P} when placing the battery on any bus.

\section{Numerical Experiment}\label{ne}
In this section, we conduct numerical experiments on different power grids to validate the effectiveness of the proposed algorithm and compare its running time to that of the Gurobi solver. We mainly study the performance of our algorithm on weakly-cyclic topology with Assumption \ref{Congestion assumption} (limited transmission) and Assumption \ref{uniform line capa} (uniform capacity). We also study the effect of Assumption \ref{unit admittance} of unit admittance by performing numerical experiements with and without this assumption. 

\begin{remark}
The Assumption \ref{uniform line capa} is hard to relax, as without it we cannot run the algorithm, which uses the homogeneous capacity $f$ in the MIF function. Therefore, we do not relax this assumption.
\end{remark}

Using the Gurobi solver, we can obtain the optimal cost both under unit and heterogeneous admittance conditions, which serves as the ground truth for cost comparison. Specifically, when the algorithm gives an optimal battery placement bus, we first verify if it matches the Gurobi solution. If consistent, the algorithm achieves optimality. Otherwise, we use Gurobi to solve the actual cost of \eqref{P} where the battery placement bus is proposed by the algorithm. This computed suboptimal cost is then compared against Gurobi's global optimum. Notably, as guaranteed by the algorithm's theoretical correctness (see appendix), this discrepancy scenario can only occur in heterogeneous admittance cases, which is confirmed by our numerical studies.

We use the following two metrics to evaluate the performance of the proposed algorithm.
\begin{itemize}
    \item[(1)] The running time of our algorithm $t_a$ against that of the solver $t_s$.
    \item[(2)] Taking $b^\prime$ as the bus proposed our algorithm. Similar to $\delta$, the metric is $\delta_{a} = (\text{opt}(b^\prime) - \text{opt}(b^*))/\text{opt}(b^*)$. In comparison, we use the relative cost difference that correlates with the ``worst" and ``mean" battery placement bus, which are $\delta_{w} = \left(\max_{i\in N}\text{opt}(i) - \text{opt}(b^*)\right)/\text{opt}(b^*)$ and $\delta_{m} = \left(\sum_{i\in N}\text{opt}(i)/N - \text{opt}(b^*)\right)/\text{opt}(b^*)$, respectively. Obviously, $0\leq \delta_{m}\leq \delta_{w}\leq \delta$. In addition, we use $(\delta_{m} - \delta_{a})^+$ to evaluate the \textit{performance increase} of our algorithm. This represents how our algorithm performs against which of randomly choose battery placement bus.
\end{itemize}
For our numerical experiments, we fix the time horizon at $T=15$. In our notation, $\mathcal{U}(l,h)$ denotes the uniform distribution with support $[l,h]$, $\text{Ber}(p)$ is the Bernoulli distribution with success probability $p$, and $\text{TN}(\mu, \sigma^2, a, b)$ represents the normal distribution with mean $\mu$ and variance $\sigma^2$ truncated to the interval $[a,b]$. 

We set $c(t) \sim \mathcal{U}(0,1)$ and $d_{i,t} \sim \mathcal{U}(1,2)$. Assuming $f=1$ to be homogeneous, we design two primary generation schemes. In Case I, the generation capacity $\tilde{g}_{i,t}$ is subject to a large variance, following $\tilde{g}_{i,t}\sim 2 + (1 - X_{i,t}) \text{TN}(0,1,0,1) + X_{i,t} (\text{TN}(0,4,-20,20) + 100)$, where $X_{i,t}\sim \text{Ber}(0.1)$. The large variance is driven by the Bernoulli distribution and the constant ``100". With 10\% probability, $X_{i,t}$ shifts the generation capacity to the high value mode. In Case II, capacity exhibits small fluctuations, given by $\tilde{g}_{i,t}\sim d_{i,t} + \mathcal{U}(0,1)$. Note that all settings satisfy our model's assumptions.


Next, we use both real and scaling power grid to evaluate Algorithm 1. As the parameters are generated randomly, we repeat the experiment 100 times for each graph.

\subsection{Numerical experiment on real power grid}
For our real-network analysis, we use topologies based on the IEEE 15-, 33-, 85-, and 123-bus test systems, modified to be weakly-cyclic. Since the 15-, 33-, and 85-bus systems are naturally tree-structured, we selectively add lines to introduce cycles. Conversely, as the IEEE 123-bus system is  not weakly-cyclic, we remove specific lines to fit the desired topology. For each system, Table \ref{tab:IEEE cases} details the modifications. We add ``M" to denote the modified bus system (see Table~\ref{tab:case_comparison_unitadd} and~\ref{tab:case_comparison_randomadd}).

\begin{table}[h]
\renewcommand\arraystretch{1.5} 
\caption{Modified IEEE bus systems for weakly-cyclic topologies.}
\label{tab:IEEE cases}
\centering
\begin{tabular}{|c|c|}
\hline
\textbf{Bus System} & \textbf{Modification} \\ \hline
IEEE 15-bus \cite{das1995simple} & \makecell[c]{(6,10), (11,14)} \\
\hline
IEEE 33-bus\cite{baran1989network}\cite{schneider2017analytic} & \makecell[c]{(1,22), (9,15), (25,29)} \\ 
\hline
IEEE 85-bus \cite{das1995simple} & \makecell[c]{(16,22), (51,71)} \\ \hline
IEEE 123-bus \cite{dolatabadi2020enhanced} & \makecell[c]{Remove (60,160)} \\ \hline
\end{tabular}
\end{table}

For each parameter setting, we evaluate Algorithm 1 under two scenarios: one with a unit admittance and another with a random admittance drawn from $\mathcal{U}(0,1)$. The results of these experiments are summarized in Tables \ref{tab:case_comparison_unitadd} and \ref{tab:case_comparison_randomadd}. Since the variance across the 100 experiments was negligible, we report only the mean values.

\begin{table}[h]
    \renewcommand\arraystretch{1.5}
    
    \begin{threeparttable}
        \caption{Test result for Algorithm 1 in real power grid with unit admittance.}
        \label{tab:case_comparison_unitadd}
        \centering
        \begin{tabular}{|c||c|c|c|c|c|c|}
            \hline
            \multicolumn{2}{|c|}{\textbf{Settings}} & $\delta_a\tnote{*}$ & $\delta_m$ & $\delta_w$ & $t_a/s$ & $t_s/s$ \\ \hline
            \multirow{2}{*}{\makecell{IEEE \\ 15-bus M}} & I & \makecell[c]{0.000\%} & \makecell[c]{10.844\%} & \makecell[c]{16.113\%} & \makecell[c]{0.003} & \makecell[c]{0.203} \\ \cline{2-7}
                & II & \makecell[c]{0.000\%} & \makecell[c]{3.779\%} & \makecell[c]{6.392\%} & \makecell[c]{0.003} & \makecell[c]{0.194} \\ \hline
            \multirow{2}{*}{\makecell{IEEE \\ 33-bus M}} & I & \makecell[c]{0.000\%} & \makecell[c]{3.500\%} & \makecell[c]{5.417\%} & \makecell[c]{0.010} & \makecell[c]{0.854} \\ \cline{2-7}
                & II & \makecell[c]{0.000\%} & \makecell[c]{1.044\%} & \makecell[c]{2.239\%} & \makecell[c]{0.010} & \makecell[c]{0.875} \\ \hline
            \multirow{2}{*}{\makecell{IEEE \\ 85-bus M}} & I & \makecell[c]{0.000\%} & \makecell[c]{1.728\%} & \makecell[c]{2.539\%} & \makecell[c]{0.040} & \makecell[c]{5.318} \\ \cline{2-7}
                & II & \makecell[c]{0.000\%} & \makecell[c]{0.571\%} & \makecell[c]{1.005\%} & \makecell[c]{0.042} & \makecell[c]{5.007} \\ \hline
            \multirow{2}{*}{\makecell{IEEE \\ 123-bus M}} & I & \makecell[c]{0.000\%} & \makecell[c]{1.439\%} & \makecell[c]{1.985\%} & \makecell[c]{0.061} & \makecell[c]{10.528} \\ \cline{2-7}
                & II & \makecell[c]{0.000\%} & \makecell[c]{0.574\%} & \makecell[c]{0.859\%} & \makecell[c]{0.065} & \makecell[c]{10.587} \\ \hline
        \end{tabular}
        \begin{tablenotes}
            \item[*] These values are zero to at least 16 decimal places.
        \end{tablenotes}
    \end{threeparttable}
\end{table}

\begin{table}[H]
    \renewcommand\arraystretch{1.5}
    
    \begin{threeparttable}
        \caption{Test result for Algorithm 1 in real power grid with random admittance.}
        \label{tab:case_comparison_randomadd}
        \centering
        \begin{tabular}{|c||c|c|c|c|c|c|}
            \hline
            \multicolumn{2}{|c|}{\textbf{Settings}} & $\delta_a$ & $\delta_m$ & $\delta_w$ & $t_a/s$ & $t_s/s$ \\ \hline
            \multirow{2}{*}{\makecell{IEEE \\ 15-bus M}} & I & \makecell[c]{0.033\%} & \makecell[c]{10.157\%} & \makecell[c]{15.355\%} & \makecell[c]{0.005} & \makecell[c]{0.326} \\ \cline{2-7}
                & II & \makecell[c]{0.154\%} & \makecell[c]{3.907\%} & \makecell[c]{6.338\%} & \makecell[c]{0.003} & \makecell[c]{0.200} \\ \hline
            \multirow{2}{*}{\makecell{IEEE \\ 33-bus M}} & I & \makecell[c]{0.031\%} & \makecell[c]{3.706\%} & \makecell[c]{5.680\%} & \makecell[c]{0.015} & \makecell[c]{1.355} \\ \cline{2-7}
                & II & \makecell[c]{0.061\%} & \makecell[c]{1.078\%} & \makecell[c]{2.207\%} & \makecell[c]{0.010} & \makecell[c]{0.895} \\ \hline
            \multirow{2}{*}{\makecell{IEEE \\ 85-bus M}} & I & \makecell[c]{0.002\%} & \makecell[c]{1.702\%} & \makecell[c]{2.482\%} & \makecell[c]{0.040} & \makecell[c]{5.156} \\ \cline{2-7}
                & II & \makecell[c]{0.000\%\tnote{\dag}} & \makecell[c]{0.585\%} & \makecell[c]{1.021\%} & \makecell[c]{0.042} & \makecell[c]{5.369} \\ \hline
            \multirow{2}{*}{\makecell{IEEE \\ 123-bus M}} & I & \makecell[c]{0.000\%\tnote{*}} & \makecell[c]{1.384\%} & \makecell[c]{1.928\%} & \makecell[c]{0.061} & \makecell[c]{10.172} \\ \cline{2-7}
                & II & \makecell[c]{0.003\%} & \makecell[c]{0.588\%} & \makecell[c]{0.881\%} & \makecell[c]{0.064} & \makecell[c]{10.476} \\ \hline
        \end{tabular}
        \begin{tablenotes}
            \item[\dag] This value is zero to 3 decimal places.
            \item[*] This value are zero to at least 16 decimal places.
        \end{tablenotes}
    \end{threeparttable}
\end{table}

As indicated by our theory, $\delta_a$ is zero in all unit admittance experiments (Table \ref{tab:case_comparison_unitadd}). A key observation from both tables is that the computational advantage (on running time) of Algorithm 1 over the Gurobi solver becomes more pronounced as the number of buses in the system increases.

Furthermore, even when the unit admittance assumption is violated (Table \ref{tab:case_comparison_randomadd}), the cost increase induced by our algorithm ($\delta_a$) remains very small. We also find that neither $\delta$ nor the running time (both for the algorithm and the solver) are sensitive to whether the admittance is unit or random. Consistent with Theorem \ref{performance bound}, the gaps $\delta_w - \delta_a$ and $\delta_m - \delta_a$ shrink in the random admittance case as the power grid scales up. These gaps are notably larger in Case I, which features more dispersed generation capacities, compared to Case II. 

Next, we evaluate the algorithm's performance on tree-structured power grids, using the original IEEE 15-bus and 85-bus systems with random admittance. To validate Remark \ref{specialtree}, we test two scenarios for line capacity: one with uniform capacity and another with random capacity, where the latter is defined in Case III.

Additionally, to verify Remark~\ref{specialtree} in a tree-structured grid, we introduce Case III. This case mirrors Case I but relaxes the homogeneity assumption by drawing the line parameter from $f \sim \text{TN}(1,1,0.01,1)$.

\begin{table}[H]
    \renewcommand\arraystretch{1.5}
    
    \begin{threeparttable}
        \caption{Test result for Algorithm 1 in tree-topology power gird with random admittance.}
        \label{tab:case_comparison_tree}
        \centering
        \begin{tabular}{|c||c|c|c|c|c|c|}
            \hline
            \multicolumn{2}{|c|}{\textbf{Settings}} & $\delta_a\tnote{*}$ & $\delta_m$ & $\delta_w$ & $t_a/s$ & $t_s/s$ \\ \hline
            \multirow{3}{*}{\makecell{IEEE \\ 15-bus}} & I & \makecell[c]{0.000\%} & \makecell[c]{11.226\%} & \makecell[c]{15.848\%} & \makecell[c]{0.002} & \makecell[c]{0.180} \\ \cline{2-7}
                & II & \makecell[c]{0.000\%} & \makecell[c]{4.524\%} & \makecell[c]{6.461\%} & \makecell[c]{0.002} & \makecell[c]{0.180} \\ \cline{2-7}
                & III & \makecell[c]{0.000\%} & \makecell[c]{7.015\%} & \makecell[c]{10.523\%} & \makecell[c]{0.003} & \makecell[c]{0.178} \\ \hline
            \multirow{3}{*}{\makecell{IEEE \\ 85-bus}} & I & \makecell[c]{0.000\%} & \makecell[c]{1.740\%} & \makecell[c]{2.522\%} & \makecell[c]{0.023} & \makecell[c]{4.989} \\ \cline{2-7}
                & II & \makecell[c]{0.000\%} & \makecell[c]{0.575\%} & \makecell[c]{0.969\%} & \makecell[c]{0.024} & \makecell[c]{5.054} \\ \cline{2-7}
                & III & \makecell[c]{0.000\%} & \makecell[c]{1.433\%} & \makecell[c]{2.097\%} & \makecell[c]{0.023} & \makecell[c]{4.985} \\ \hline
        \end{tabular}
        \begin{tablenotes}
            \item[*] These values are zero to at least 16 decimal places.
        \end{tablenotes}
    \end{threeparttable}
\end{table}
As shown in Table \ref{tab:case_comparison_tree}, similar trends hold for tree topologies. The key finding is that $\delta_a=0$ in Case III—even with random admittance and heterogeneous line capacity—which provides direct empirical support for Remark \ref{specialtree}.

Similar results were obtained following other modifications to the original IEEE systems and the implementation of other randomized parameter generation schemes, shown in appendix.

\subsection{Numerical experiment on scaling power grid}
The phenomenon that the gaps $\delta_w - \delta_a$ and $\delta_m - \delta_a$ shrink from 15- to 123-bus system motivates us to design a simulation in which the power grid scales up gradually. We construct a series of artificial weakly-cyclic power grids with triangle structures.

We ran simulations on power grids ranging in size from 15 to 55 buses. For each grid size, we test two distinct parameter generation schemes, Case I and Case II. Each of these experimental settings (a specific grid size and case) is simulated 100 times. Moreover, we observe that the solver's running time is larger than the algorithm. This motivates us to introduce a benchmark setting the solver's time limit comparable to that of our algorithm, aiming to observe an improvement in accuracy from our algorithm over the solver within a limited time. However, in our case, setting a maximum solving time comparable to that of the algorithm is insufficient for the solver to find the optimal solution. As a result, for each experiment, we first run Algorithm 1 and then impose a solving time constraint on the solver equal to 10 times the solving time of Algorithm 1 (excluding modeling time). The corresponding cost increase is denoted as $\delta_t$. The running time and $\delta$ are depicted in Fig. \ref{fig:running-time1} and \ref{fig:running-time2}.

\begin{figure}[h]
    \centering
    \includegraphics[width=1\linewidth]{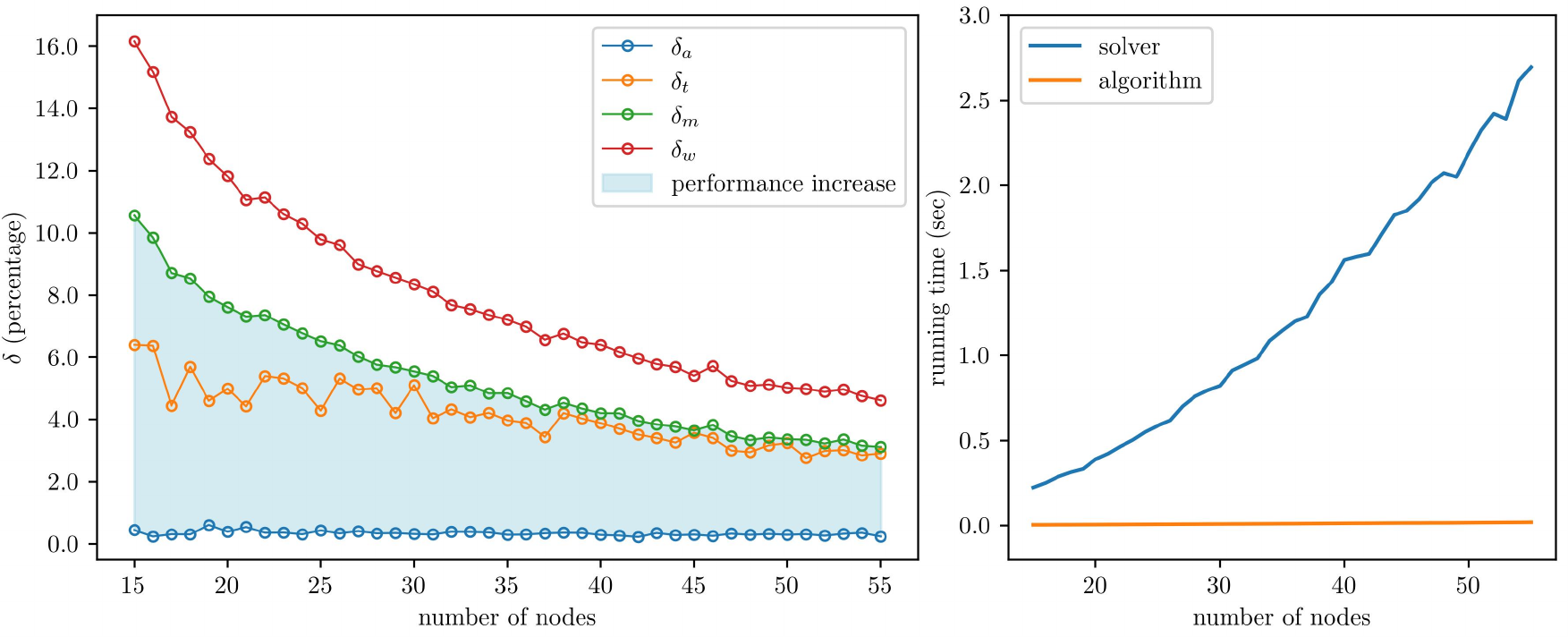}
    \caption{$\delta$s and running time comparison: solver and algorithm (Case I).}
    \label{fig:running-time1}
\end{figure}

\begin{figure}[h]
    \centering
    \includegraphics[width=1\linewidth]{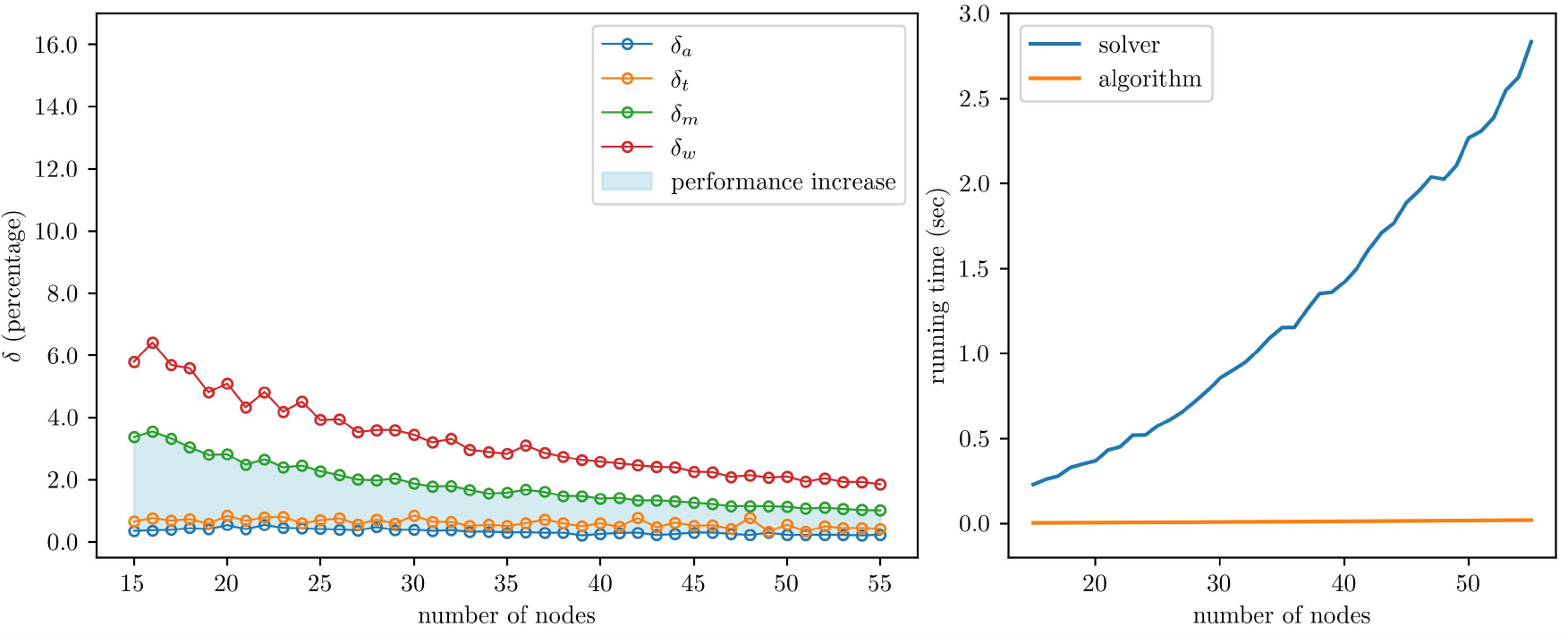}
    \caption{Loss and running time comparison: solver and algorithm (Case II).}
    \label{fig:running-time2}
\end{figure}

The preceding results lead to several key conclusions. First, our algorithm demonstrates a significant computational advantage, running approximately 100 times faster on average than the Gurobi solver. Second, the algorithm proves to be robust: when applied to weakly-cyclic graphs with heterogeneous admittance, it incurs only a small increase in cost. Finally, cost increases ($\delta$) diminishe as the network size grows and are consistently lower for systems with more centralized generation capacities (Case II) compared to those with more dispersed capacities (Case I).

\section{Conclusion}\label{conclusion}
In this study, we investigate the optimal placement of a single, infinitely large capacity battery in a power grid to minimize total generation cost, with a primary focus on finding the influence of network topology on this decision.

Our main contribution begins with the decomposition of the original MILP problem into a set of independent LP subproblems, one for each potential battery location. By introducing the Maxflow decomposition, we are able to characterize the optimal battery placement as a bang-bang control and derive an analytical expression for the operational cost.

The theoretical analysis reveals that in scenarios with sufficiently large generation capacity or in a tree network, the weighted degree is the only determinant of optimal battery placement under some parameter settings. In tree networks, the minor impact of higher-order metrics such as betweenness centrality is also discussed. Furthermore, we establish a performance bound indicating that the relative cost-saving benefit of a single battery diminishes as the network scales, underscoring the localized nature of its impact on large systems.

Building on these theoretical foundations, we propose Algorithm 1 to efficiently solve for the optimal placement in weakly-cyclic graphs. Numerical experiments conducted on modified IEEE test systems and scaled artificial grids validate our approach. The algorithm not only showcases a dramatic improvement in computational efficiency, but is also precise even when the unit admittance assumption is removed.


Future research could extend this framework in several promising directions. Generalizing our analytical results and the proposed algorithm to arbitrary network topologies remains a key challenge. Additionally, investigating the problem with multiple interacting battery units, incorporating stochastic models for demand and generation, and accounting for transmission losses would further enhance the model's practical applicability in real-world energy systems.

\bibliographystyle{IEEEtran}



\end{document}